\documentclass{article}
\usepackage[english]{babel}

\usepackage[a4paper,top=2cm,bottom=2cm,left=3cm,right=3cm,marginparwidth=1.75cm]{geometry}

\usepackage{mathtools} 
\usepackage{amsxtra,amssymb,amsthm,amstext} 

\usepackage[colorlinks=true, allcolors=blue]{hyperref}
\usepackage[
    style=authoryear, 
    url=false,
    giveninits=true, 
    uniquename=false, 
    uniquelist=false, 
    maxcitenames=2,   
    maxbibnames=99,    
    dashed=false  
]{biblatex}

\usepackage{authblk}

\renewcommand{\cite}[1]{\parencite{#1}}
\renewcommand{\d}{\!\mathrm{d}}
\newcommand{\e}{\mathrm{e}}
\newcommand{\mean}[1]{\mathbb{E}\left(#1\right)}

\newtheorem{theorem}{Theorem}[section]
\newtheorem{lemma}{Lemma}[section]
\newtheorem{definition}{Definition}[section]

\title{Information on hidden birth events restores identifiability in phylodynamic inference}
\author[1,2]{Tobias Dieselhorst\thanks{tobias.dieselhorst@bsse.ethz.ch}}
\author[1,2]{Tanja Stadler}
\affil[1]{Department of Biosystems Science and Engineering, ETH Zürich, 4056 Basel, Switzerland}
\affil[2]{Swiss Institute of Bioinformatics, 1015 Lausanne, Switzerland}
\date{}

\begin{document}
\maketitle

\begin{abstract}
The parameters of many classes of birth-death processes cannot be inferred uniquely from phylogenetic trees: infinitely many parameter combinations yield the same distribution of phylogenetic trees. Here, we show that parameter identifiability can be recovered even for the most general cases of time-dependent rates when additional information on hidden birth events along branches of the reconstructed tree is available. This holds both for models in which individuals are sampled at a single point in time or through time at a time-dependent rate. Moreover, we prove that when mutations occur at birth -- assuming two different models for the accumulation of mutations at a birth event -- then information about hidden birth events is available in the sequences and thus all parameters of time-dependent birth-death models become identifiable. Thus, phylodynamic inference is identifiable whenever evolutionary models with mutation accumulation at birth (such as at speciation, transmission, or cell division) are plausible.
\end{abstract}

\section{Introduction}
Phylodynamic inference, the study of the population dynamics behind phylogenetic trees, has found wide applications in the analysis of the evolution of species, epidemics, and populations of cells \cite{stadler_phylodynamics_2021,macpherson_unifying_2022}. While phylogenetic tree reconstruction and phylodynamic inference are often done in parallel in Bayesian statistical analysis software \cite{hohna_revbayes_2016,bouckaert_beast_2019,baele_beast_2025}, the evolutionary process of mutation (or substitution; note that in what follows we write mutations which may correspond also to a substitution pending on application context) accumulation in sequences and the tree generating process are typically modeled independently.

Models of molecular evolution are the basis for the phylogenetic likelihood, the probability of a sequence alignment given a phylogenetic tree (with branch lengths corresponding to calendar time). They typically involve molecular clocks, where mutations in the genetic sequence of individuals occur independently at a rate over time and are therefore independent of the population dynamics. The mutations used for phylogenetic reconstruction are usually assumed to be neutral, and thus, the population dynamics remain unaffected.

Population dynamics are commonly modeled with birth-death processes. Here, individuals of a population replicate or die independently with birth rate $\lambda$ and death rate $\mu$. The case of constant rates was first introduced by Feller \cite{feller_grundlagen_1939}. The phylodynamic likelihood of such a model, i.e. the probability distribution over phylogenetic trees under a given model, also known as the tree prior, has been established for various sampling scenarios. The probability of a phylogenetic tree reconstructed from extant individuals which are sampled independently with probability $\rho$ at a single point in time has been established in \cite{stadler_incomplete_2009}. Birth-death-sampling models exist too, where individuals are sampled through time with a sampling rate $\psi$ \cite{stadler_sampling-through-time_2010}. Here, we only consider the case where sampled individuals are subsequently removed from the population.

The generalized birth-death model introduced by Kendall \cite{kendall_generalized_1948} allows the birth and death rates to vary through time. The functions of time for each rate are also called rate trajectories, and their likelihood under a given phylogenetic tree has been established for both previously described sampling scenarios \cite{morlon_reconciling_2011,macpherson_unifying_2022}. Detecting time-heterogeneity in the population dynamics is of particular interest in many applications. However, Louca, Pennell and collaborators have shown that the distribution of phylogenetic trees is not unique for each set of model parameters \cite{louca_extant_2020,louca_fundamental_2021}: Infinitely many models, which they term a congruence class, yield the same distribution of phylogenetic trees and, therefore, cannot be distinguished even in the case of infinitely large or many trees. This so-called non-identifiability can lead to misleading results, in particular when restricting the inference of rate trajectories to a particular class of functions.

Here, we show that these congruence classes resolve if we have information on hidden birth events: birth events are either visible in the phylogenetic tree, if both descendants leave sampled offspring, or lie hidden along branches, if one of the two descendants does not leave any sampled offspring. Genetic sequences do not contain information on hidden birth events if mutations occurred independent of the population dynamics.

Recently, however, models have been proposed that couple molecular evolution to birth events \cite{manceau_model-based_2020,douglas_evolution_2025,dieselhorst_phylodynamics_2026}. We will refer to these models as burst models. In single-cell phylodynamics, such models arise naturally when mutations occur at cell divisions. But punctuated evolution at branching events has also been proposed for species and language evolution \cite{eldredge_punctuated_1972,atkinson_languages_2008}. 
Hidden birth events contribute to the number of mutations along a branch. As the number of hidden birth events along branches depends on the birth-death process, burst models couple molecular evolution to the population dynamics. Based on these modeling assumptions, it seems very plausible that the resulting sequencing data contains information about the hidden birth events in a phylogeny and thus that identifiability is ensured. 

We will show that indeed this coupling allows us to recover identifiability of generalized birth-death process from the phylogenetic tree and the sequence alignment. We begin by showing that knowing the number of hidden birth events along branches of the phylogenetic tree makes these models fully identifiable. We prove this result for both sampling schemes and combinations of the two. For two established burst models, we then show that the sequencing data contains sufficient information on hidden birth events to identify both the birth-death process and the parameters of the burst model. The proofs of these results are simple and can easily be repeated for novel burst models in the future.

\section{Results}
We begin by defining identifiability.
\begin{definition}[Identifiability]
	We say the parameters $\theta$ of a model are identifiable (or simply the model is identifiable) from a likelihood function $P(D\mid \theta)$ if different parameters $\theta$ and $\theta'$ always yield different likelihoods for all possible data $D$, i.e.\ $P(D\mid \theta) = P(D\mid \theta')$ for all $D$ implies $\theta=\theta'$. This corresponds to the mapping from model parameters $\theta$ to the distribution of data $P(D\mid \theta)$ being injective.
\end{definition}

In the following, we will derive that all parameters of time-dependent birth-death models are identifiable from phylogenetic trees and the number of hidden birth events along branches. For simplicity, we will show that having both the same distributions of phylogenetic trees and of the number of hidden birth events along branches (rather then the same joint distribution) implies that two models are identical. The following lemma shows that this is sufficient to prove identifiability from the model likelihood, i.e.\ the joint probability of the tree $\mathcal{T}$ and the number of hidden birth events $\mathcal{B}=\left\{B_i\right\}_{i=1,...,n_{\mathcal{T}}}$ along each of the $n_{\mathcal{T}}$ branches of the tree
\begin{align}
    P(\mathcal{T},\mathcal{B}\mid \theta) &= P(\mathcal{T}\mid \theta)P(\mathcal{B}\mid \mathcal{T},\theta) \nonumber \\
    &= P(\mathcal{T}\mid \theta)\prod_{i=1}^{n_{\mathcal{T}}}P(B_i\mid \mathcal{T},\theta).
\end{align}
where the last equality follows from the independence of the number of hidden birth events along each branch.
\begin{lemma}
	We say a model with parameters $\theta$ is identifiable from the distributions $P(X\mid \theta)$ and $P(Y\mid X,\theta)$ of two random variables $X$ and $Y$ (e.g.\ trees $X$ and hidden birth events $Y$) if the mapping form $\theta$ to the tuple of distributions $(P(X\mid \theta), P(Y\mid X,\theta))$ is injective. In this case, the model is also identifiable from the joint distribution $P(X,Y\mid \theta)= P(Y\mid X,\theta)P(X\mid \theta)$.
    \label{lem:joint}
\end{lemma}
\begin{proof}
	We assume two sets of parameters $\theta$ and $\theta'$ that yield the same joint distribution of $X$ and $Y$, i.e.
	\begin{align*}
		P(Y\mid X,\theta)P(X\mid \theta) = P(Y\mid X,\theta')P(X\mid \theta')
	\end{align*}
	for all $X$ and $Y$. Integrating both sides over $Y$, we get that \begin{align}
	    P(X\mid \theta) = P(X\mid \theta') \label{eq:ident_X}
	\end{align} for all $X$ and thus $P(Y\mid X,\theta)P(X\mid \theta) = P(Y\mid X,\theta')P(X\mid \theta)$.
	From this, we deduce that 
    \begin{align}
        P(Y\mid X,\theta) = P(Y\mid X,\theta') \label{eq:ident_Y}
    \end{align} for all $X$ and $Y$.
	If the parameters are identifiable from $(P(X\mid \theta),P(Y\mid X,\theta))$, results \eqref{eq:ident_X}-\eqref{eq:ident_Y} imply that $\theta=\theta'$. Thus, the parameters are also identifiable from the joint distribution $P(X,Y\mid \theta)$.
\end{proof}

\subsection{Identifiability from trees with extant sampling and hidden birth events}
\label{sec:extant_sampling}
Time-dependent birth-death models with extant sampling are defined by their birth and death rate trajectories $\lambda(\tau)$ and $\mu(\tau)$, where $\tau\geq0$ is the time before the observation (and thus running backwards). At $\tau=0$, each extant individual is sampled independently with probability $\rho$, and the phylogenetic tree of these samples is built.
Louca and Pennell \cite{louca_extant_2020} have shown that various parameter combinations yield the same distribution of phylogenetic trees, i.e.\ are congruent. This is true even when the sampling probability $\rho$ is known.

By showing that different models never have the same distributions of both phylogenetic trees and the number of hidden birth events along branches, we prove that they are fully identifiable when the number of hidden birth events is known.
\begin{theorem}
	Let $\lambda(\tau),\mu(\tau),\lambda'(\tau),\mu'(\tau)$ be strictly positive and continuously differentiable functions of time $\tau\geq 0$ and $\rho,\rho'\in(0,1]$ be constants.
	If two time-dependent birth-death models with extant sampling $\theta_{BD}=(\lambda(\tau),\mu(\tau),\rho)$ and $\theta'_{BD}=(\lambda'(\tau),\mu'(\tau),\rho')$ yield the same distributions of phylogenetic trees and the same distribution for the number of hidden birth events along branches, then they are identical, i.e.\ $\lambda(\tau)=\lambda'(\tau)$, $\mu(\tau)=\mu'(\tau)$ for all $\tau \geq 0$ and $\rho=\rho'$.
	\label{th:time_dep}
\end{theorem}

\begin{proof}
The number of hidden birth events $B^{(\tau_s,\tau_e)}$ along a branch from time $\tau_s$ to $\tau_e<\tau_s$ before the present is Poisson distributed with mean 
\begin{align}
	\mean{B^{(\tau_s,\tau_e)}\mid\theta_{BD}} = \int_{\tau_e}^{\tau_s}\d\tau 2\lambda(\tau)p_0(\tau)
	\label{eq:mean_B_time_dep}
\end{align}
(see Theorem~\ref{th:poisson_B_extant_sampling} in the Appendix).
Here, $p_0(\tau)$ denotes the probability of an individual alive at time $\tau$ leaving no sampled descendants and is given by
\begin{align}
	p_0(\tau) = 1 - \frac{\e^{\int_0^\tau\d u\left[ \lambda(u)-\mu(u) \right]}}{\frac{1}{\rho} + \int_0^\tau\d s\e^{\int_0^s\d u\left[ \lambda(u)-\mu(u) \right]}\lambda(s)}
	\label{eq:p0_extant}
\end{align}
\cite{morlon_reconciling_2011}.
According to \cite{louca_extant_2020}, an alternative model $\theta'_{BD}=(\lambda'(\tau),\mu'(\tau),\rho')$ is congruent if and only if
\begin{align}
	\rho\lambda(0) &= \rho'\lambda'(0) \label{eq:rho_lambda_0} \\
	\lambda(\tau)-\mu(\tau)+\frac{1}{\lambda(\tau)}\frac{\d\lambda(\tau)}{\d\tau} &= \lambda'(\tau)-\mu'(\tau)+\frac{1}{\lambda'(\tau)}\frac{\d\lambda'(\tau)}{\d\tau}.
	\label{eq:pulled_div_rate}
\end{align}
The latter quantity is also called the pulled diversification rate \cite{louca_bacterial_2018,louca_extant_2020}. The exponentials in equation~\eqref{eq:p0_extant} of two congruent models are thus related according to
\begin{align*}
	\int_0^\tau\d u\left[ \lambda'(u)-\mu'(u) \right]
	&= \int_0^\tau\d u\left[ \lambda(u)-\mu(u)+\frac{1}{\lambda(u)}\frac{\d\lambda(u)}{\d u}-\frac{1}{\lambda'(u)}\frac{\d\lambda'(u)}{\d u} \right] \\
	&= \int_0^\tau\d u\left[ \lambda(u)-\mu(u) \right] + \log{\frac{\lambda(\tau)}{\lambda(0)}\frac{\lambda'(0)}{\lambda'(\tau)}} \\
	&= \int_0^\tau\d u\left[ \lambda(u)-\mu(u) \right] + \log{\frac{\rho\lambda(\tau)}{\rho'\lambda'(\tau)}},
\end{align*}
where we have used equality~\eqref{eq:rho_lambda_0} in the last step.
By inserting this result into the expression of the probability of extinction~\eqref{eq:p0_extant}, we find that for two congruent models, the latter are related by
\begin{align*}
	p'_0(\tau) &= 1 - \frac{\frac{\rho\lambda(\tau)}{\rho'\lambda'(\tau)}\e^{\int_0^\tau\d u\left[ \lambda(u)-\mu(u) \right]}}{\frac{1}{\rho'} + \frac{\rho}{\rho'}\int_0^\tau\d s \e^{\int_0^s\d u\left[ \lambda(u)-\mu(u) \right]}\lambda(s)} \\
	&= 1 - \frac{\frac{\lambda(\tau)}{\lambda'(\tau)}\e^{\int_0^\tau\d u\left[ \lambda(u)-\mu(u) \right]}}{\frac{1}{\rho} + \int_0^\tau\d s \e^{\int_0^s\d u\left[ \lambda(u)-\mu(u) \right]}\lambda(s)} \\
	&= \frac{\lambda(\tau)}{\lambda'(\tau)}\left(p_0(\tau)-1\right) + 1.
\end{align*}
Here, $p'_0(\tau)$ denotes the probability of extinction for the alternative model $\theta'_{BD}$.
Plugging this into equation~\eqref{eq:mean_B_time_dep}, we find
\begin{align*}
	\mean{B^{(\tau_s,\tau_e)}\mid\theta'_{BD}} &= \int_{\tau_e}^{\tau_s}\d\tau 2\lambda'(\tau) p'_0(\tau) \\
	&= \int_{\tau_e}^{\tau_s}\d\tau2\lambda(\tau)p_0(\tau) + \int_{\tau_e}^{\tau_s}\d\tau2\left[\lambda'(\tau) - \lambda(\tau)\right] \\ 
	&= \mean{B^{(\tau_s,\tau_e)}\mid\theta_{BD}} + 2\int_{\tau_e}^{\tau_s}\d\tau\left[\lambda'(\tau) - \lambda(\tau)\right] .
\end{align*}
Thus, the mean number of hidden birth events \eqref{eq:mean_B_time_dep} is equal for both processes if and only if \linebreak $\int_{\tau_e}^{\tau_s}\d\tau\left[\lambda'(\tau)-\lambda(\tau)\right]=0$. In order for two models to yield the same expected number of hidden birth events along any branch of any tree, this integral must vanish for all $\tau_s$ and $\tau_e$. This is only possible if $\lambda'(\tau)=\lambda(\tau)$ for all $\tau\geq 0$. Then, from equation~\eqref{eq:rho_lambda_0}, we deduce that the sampling probabilities $\rho$ and $\rho'$ must be the same as well. Combining these results with the equality of the pulled diversification rates~\eqref{eq:pulled_div_rate}, we find that the death rates must be equal everywhere as well, i.e.\ $\mu(\tau)=\mu'(\tau)$. As Poisson distributions are uniquely defined by their means, this concludes the proof.
\end{proof}

\subsubsection*{Implications for constant-rate models}
Even in the case of constant birth and death rates, it is known that only two of the three parameters $(\lambda,\mu,\rho)$ can be inferred from a phylogenetic tree alone \cite{stadler_incomplete_2009,stadler_how_2013}. Theorem~\ref{th:time_dep} implies that when the number of hidden birth events along branches is known, all three parameters become identifiable. In Appendix~\ref{sec:const_rates_derivation}, we provide an alternative direct proof of identifiability for this particular case.

\subsection{Identifiability from trees with sampling through time and hidden birth events}
\label{sec:sampling_through_time}
Congruence classes also exist for time-dependent birth-death processes with samples taken at rate $\psi(\tau)$ through time \cite{louca_fundamental_2021}, if samples are removed from the population \cite{truman_fossilized_2025}. As above, we show that congruent models with the same statistics of hidden birth events are identical.

\begin{theorem}
	Let $\lambda(\tau),\mu(\tau),\psi(\tau),\lambda'(\tau),\mu'(\tau),\psi'(\tau)$ be strictly positive continuously differentiable functions of time $\tau\geq 0$.
	If two time-dependent birth-death-sampling models $\theta_{BDS}=(\lambda(\tau),\mu(\tau),\psi(\tau))$ and $\theta'_{BDS}=(\lambda'(\tau),\mu'(\tau),\psi'(\tau))$ yield the same distribution of phylogenetic trees and the same distribution for the number of hidden birth events along branches, then they are identical, i.e.\ $\lambda(\tau)=\lambda'(\tau)$, $\mu(\tau)=\mu'(\tau)$ and $\psi(\tau)=\psi'(\tau)$ for all $\tau \geq 0$.
	\label{th:sampling_through_time}
\end{theorem}

\begin{proof}
Two birth-death-sampling models are congruent if and only if
\begin{align}
	(1-p_0(\tau))\lambda(\tau) &= (1-p'_0(\tau))\lambda'(\tau) \label{eq:pulled_birth_rate} \\
	\lambda(\tau)-\mu(\tau)-\psi(\tau) + \frac{1}{\lambda(\tau)}\frac{\d\lambda(\tau)}{\d\tau} &= \lambda'(\tau)-\mu'(\tau)-\psi'(\tau) + \frac{1}{\lambda'(\tau)}\frac{\d\lambda'(\tau)}{\d\tau}\label{eq:pulled_div_rate_sampling}
\end{align}
where the first quantity is called the pulled birth rate and the second one the pulled diversification rate with sampling \cite{louca_fundamental_2021}.
Again, $p_0(\tau)$ denotes the probability of a lineage not being in the phylogenetic tree, and is given as the solution of the differential equation
\begin{align}
	\frac{\d p_0(\tau)}{\d\tau} &= - \left(\lambda(\tau)+\mu(\tau)+\psi(\tau)\right)p_0(\tau) + \lambda(\tau) p_0(\tau)^2 + \mu(\tau) \label{eq:diff_eq_p0} \\
	p_0(0) &= 1 \label{eq:init_p0}
\end{align}
\cite{louca_fundamental_2021,macpherson_unifying_2022}.

As for the case of extant sampling, the number of hidden birth events $B^{(\tau_s,\tau_e)}$ along a branch from $\tau_s$ to $\tau_e$ is Poisson distributed with mean
\begin{align*}
	\mean{B^{(\tau_s,\tau_e)}\mid\theta_{BDS}} = \int_{\tau_e}^{\tau_s}\d\tau 2\lambda(\tau)p_0(\tau)
\end{align*}
(see Theorem~\ref{th:poisson_B_sampling_through_time} in the Appendix).

For two congruent models $\theta_{BDS}=(\lambda(\tau),\mu(\tau),\psi(\tau))$ and $\theta'_{BDS}=(\lambda'(\tau),\mu'(\tau),\psi'(\tau))$, we find
\begin{align*}
	\mean{B^{(\tau_s,\tau_e)}\mid\theta'_{BDS}} &= \int_{\tau_e}^{\tau_s}\d\tau 2\lambda'(\tau)p'_0(\tau) \\
	&= 2\int_{\tau_e}^{\tau_s}\d\tau \left[\lambda'(\tau) - \left(1-p_0(\tau)\right)\lambda(\tau)\right] \\
	&= 2\int_{\tau_e}^{\tau_s}\d\tau \lambda(\tau)p_0(\tau) + 2\int_{\tau_e}^{\tau_s}\d\tau \left[\lambda'(\tau) - \lambda(\tau)\right] \\
	&= \mean{B^{(\tau_s,\tau_e)}\mid\theta_{BDS}} + 2\int_{\tau_e}^{\tau_s}\d\tau \left[\lambda'(\tau) - \lambda(\tau)\right]
\end{align*}
where in the second line, we have used equality~\eqref{eq:pulled_birth_rate}. Applying the same arguments as in the proof of Theorem~\ref{th:time_dep}, the two models will only have the same expected number of hidden birth events $\mean{B^{(\tau_s,\tau_e)}\mid\theta_{BDS}}=\mean{B^{(\tau_s,\tau_e)}\mid\theta'_{BDS}}$ if $\lambda(\tau)=\lambda'(\tau)$ for all $\tau\geq 0$.
Plugging the equality of birth rates into the equality of pulled birth rates \eqref{eq:pulled_birth_rate}, we find that both models must also have the same extinction probabilities at all times $\tau$, i.e. $p_0(\tau) = p'_0(\tau)$.
Furthermore, inserting this result into the equality of the pulled diversification rates~\eqref{eq:pulled_div_rate_sampling} for congruent models, we find that the sum of death and sampling rates must be equal at all times too, i.e. $\mu(\tau) + \psi(\tau) = \mu'(\tau) + \psi'(\tau)$ for all $\tau\geq 0$.
From these equalities and the differential equation~\eqref{eq:diff_eq_p0} for $p_0(\tau)$, in which the last term only depends on $\mu(\tau)$, we deduce that $\mu(\tau)=\mu'(\tau)$ for all $\tau\geq 0$, and therefore also $\psi(\tau)=\psi'(\tau)$. This concludes the proof, as the Poisson distribution is uniquely defined by its mean, and we have shown that for congruent models $\theta_{BDS}$ and $\theta'_{BDS}$, $\mean{B^{(\tau_s,\tau_e)}\mid\theta_{BDS}}=\mean{B^{(\tau_s,\tau_e)}\mid\theta'_{BDS}}$ for all $\tau_s$ and $\tau_e$ implies that the models are identical, i.e.\ $\theta_{BDS}=\theta'_{BDS}$.
\end{proof}

\subsubsection*{Sampling through time and at the present}
It is straightforward to extend the results of Theorems \ref{th:time_dep} and \ref{th:sampling_through_time} to combined models $\theta_{BDC}=(\lambda(\tau),\mu(\tau),\psi(\tau),\rho)$ where individuals are sampled through time with rate $\psi(\tau)$ and at the present with probability $\rho$: A tree $\mathcal{T}_{h,BDC}$ arising from such a combined model (denoted by the $BDC$ subscript) with number of hidden birth events along each branch known (denoted by the $h$ subscript) can be decomposed into two trees - $\mathcal{T}_{h,BDS}$ containing only the samples through time and $\mathcal{T}_{h,BD}$ containing only the samples of extant individuals at the present (again, the number of hidden birth events along the branches of these trees are known, thus the subscript $h$). Even though this decomposition reduces the amount of information, we know that each of the original models is identifiable from its respective tree, i.e.\ the mappings from $\theta_{BD}=(\lambda(\tau),\mu(\tau),\rho)$ to $P(\mathcal{T}_{h,BD}\mid\theta_{BD})$ and from $\theta_{BDS}=(\lambda(\tau),\mu(\tau),\psi(\tau))$ to $P(\mathcal{T}_{h,BDS}\mid\theta_{BDS})$ are injective. Therefore, the same is true for the mapping from $(\theta_{BD},\theta_{BDS})$ to $(P(\mathcal{T}_{h,BD}\mid\theta_{BD}),P(\mathcal{T}_{h,BDS}\mid\theta_{BDS}))$. Due to the injectivity of the mapping from $\theta_{BDC}$ to $(\theta_{BD},\theta_{BDS})$ and the composition of injective functions being injective, $\theta_{BDC}$ is identifiable from $(P(\mathcal{T}_{h,BD}\mid\theta_{BD}),P(\mathcal{T}_{h,BDS}\mid\theta_{BDS}))$.

\subsection{Information in sequencing data on hidden birth events}
\label{sec:summing_over_B}
In practice, the exact number of hidden birth events along branches is often unknown. Instead, we may only have some information on the number of hidden birth events along a branch. This information can be expressed as the probability $P(D\mid B, M_\mathrm{burst})$ of observing data $D$ along a branch with $B$ hidden birth events and a burst model $M_\mathrm{burst}$ for the relationship between $B$ and $D$. In this context, data $D$ could be the number of mutations. \cite{dieselhorst_phylodynamics_2026} or evolutionary distance \cite{douglas_evolution_2025} along a branch, if these measures increase at birth events. We then assume $D$ to be the sum of independent and identically distributed contributions $\hat{D}_i$ from each of the $B$ hidden birth events and one additional birth event from which the branch originates, i.e. $D = \sum_{i=1}^{B+1} \hat{D}_i$.
To compute the probability density of $D$ along a branch under a given birth-death model $\theta$ (either of the ones discussed above) and a burst model $M_\mathrm{burst}$, we sum over all possible (Poisson distributed) numbers of hidden birth events
\begin{align}
	P(D=d\mid \theta,M_\mathrm{burst}) &= \sum_{b=0}^\infty P\left(D\mid B=b,M_\mathrm{burst}\right) \frac{\e^{-\beta}\beta^b}{b!} \nonumber \\ 
	&= \sum_{b=0}^\infty P\left(\sum_{i=1}^{b+1}\hat{D}_i=d\mid M_\mathrm{burst}\right) \frac{\e^{-\beta}\beta^b}{b!}.
	\label{eq:shifted_comp_poisson}
\end{align}
This is a compound Poisson distribution \cite{dieselhorst_phylodynamics_2026} with a shift to account for the observed birth event at the beginning of a branch (the plus one in the upper bound of the nested sum). $\beta\equiv\mean{B^{(\tau_s,\tau_e)}\mid \theta}$ is the expected number of hidden birth events along the branch considered.
If $\beta$ is identifiable from this distribution, we can use the theorems above to prove identifiability of the entire birth-death process $\theta$ from \eqref{eq:shifted_comp_poisson} and the distribution of trees.

The exact form of the distribution of $\hat{D}$ is usually unknown. Therefore, and because of the shift in \eqref{eq:shifted_comp_poisson}, standard results on the identifiability of mixture distributions \cite{teicher_identifiability_1961} cannot be applied. However, for many classes of parametrized distributions for $\hat{D}$, one can show that both, the parameters defining the distribution of $\hat{D}$ (and thus characterizing $M_\mathrm{burst}$) and $\beta$, are identifiable from \eqref{eq:shifted_comp_poisson}.

It turns out that identifiability can be shown conveniently from the cumulants. The moment generating function (MGF) and cumulant generating function (CGF) of $D$ are given by
\begin{align}
	M(t) &= \mean{\e^{tD}} = M_{\hat{D}}(t)\e^{\beta\left(M_{\hat{D}}(t)-1\right)} \label{eq:mom_gen_D} \\
    K(t) &= \log{M(t)} = K_{\hat{D}}(t) + \beta\left(M_{\hat{D}}(t)-1\right), \label{eq:cum_gen_D}
\end{align}
with $M_{\hat{D}}(t)$ denoting the MGF and $K_{\hat{D}}(t)$ the CGF of $\hat{D}$.
The $n$-th cumulant of the distribution \eqref{eq:shifted_comp_poisson} of $D$ is given by the $n$-th derivative of the CGF \eqref{eq:cum_gen_D} evaluated at $t=0$ and therefore given by
\begin{align}
    \kappa_n &= \kappa_{\hat{D},n} + \beta m_{\hat{D},n}, \label{eq:cumulant_D}
\end{align}
where $\kappa_{\hat{D},n}$ is the $n$-th cumulant and $m_{\hat{D},n}$ is the $n$-th moment of $\hat{D}$.
Note that if a cumulant generating function exists (in a neighborhood around $t=0$), then they uniquely define the probability distribution. Thus, the strategy to show identifiability is to prove that different parameters yield different cumulants or cumulant generating functions and therefore different distributions of $D$.
As examples, we now show the identifiability of two previously established burst models.

We begin with the case considered in \cite{dieselhorst_phylodynamics_2026}, where $\hat{D}$ are Poisson distributed mutation counts. We note that, while the model was introduced with constant rates and extant sampling only, we prove identifiability for any of the birth-death models considered above.
\begin{theorem}
	Let the $\hat{D}$ be Poisson distributed with mean $\eta>0$, $B$ be Poisson distributed with mean $\beta>0$ and $D=\sum_{i=0}^{B+1}\hat{D}_i$. Then both, $\eta$ and $\beta$ are identifiable from the shifted compound Poisson distribution~\eqref{eq:shifted_comp_poisson} of $D$.
\end{theorem}
\begin{proof}
The first three moments of a Poisson distribution with mean $\eta$ are given by
\begin{align*}
    m_{\hat{D},1} &= \eta \\
    m_{\hat{D},2} &= \eta + \eta^2 \\
    m_{\hat{D},3} &= \eta + 3\eta^2 + \eta^3
\end{align*}
and all cumulants are equal to $\eta$. Using equation \eqref{eq:cumulant_D}, we find the first three cumulants of the distribution of $D$
\begin{align}
	\kappa_1 &= (\beta+1) \eta \label{eq:k1_poi} \\
	\kappa_2 &= \beta\eta^2 + (\beta + 1)\eta \label{eq:k2_poi} \\
    \kappa_3 &= \beta\eta^3 + 3\beta\eta^2 + \beta\eta + \eta.\label{eq:k3_poi}
\end{align}
Solving the equations for the first two cumulants \eqref{eq:k1_poi}-\eqref{eq:k2_poi} for $\eta$ and $\beta$ yields two possible solutions
\begin{align*}
	\eta_\pm &= \frac{1}{2} \left( \kappa_1 \pm \sqrt{\kappa_1^2 + 4\left(\kappa_1-\kappa_2\right)}  \right) \\
	\beta_\pm &= \frac{\eta_\mp}{\eta_\pm}.
\end{align*}
The two solutions are identical if the square root vanishes. For the general case where the two solutions are not the same, we use equation \eqref{eq:k3_poi} to compute the corresponding third cumulants
\begin{align*}
    \kappa_{3,\pm} &= \eta_\mp\eta_\pm^2 + 3\eta_\mp\eta_\pm + \eta_\mp + \eta_\pm \\
    &= - (\kappa_1-\kappa_2)\eta_\pm - (\kappa_1-\kappa_2) + \kappa_1.
\end{align*}
We now assume that both solutions yield the same third cumulant, i.e. $\kappa_{3,+}=\kappa_{3,-}$. However, from this and the equality of the first two cumulants would follow $\eta_+=\eta_-$ and therefore also $\beta_+=\beta_-$. Thus, by contradiction, we have proven that every set of parameters $(\eta,\beta)$ yields different cumulants $(\kappa_1,\kappa_2,\kappa_3)$ or, in other words, both parameters are uniquely identifiable from the first three cumulants and, therefore, also from the distribution of $D$.
\end{proof}

Next, we consider the gamma spike model \cite{douglas_evolution_2025}, where $D$ is a positive real-valued evolutionary distance defined by $D=s_\mu\times S_B$ where $s_\mu$ is a positive parameter and $S_B$ is a random variable following a gamma distribution $\Gamma(\alpha=s_\alpha(B+1),\theta=1/s_\alpha)$ ($\alpha$ us the shape and $\theta$ the scale parameter). The sum of $n$ independent random variables with distribution $\Gamma(\alpha,\theta)$ is itself gamma distributed with $\Gamma(n\alpha,\theta)$. Therefore, the contribution of each birth event is given by $\hat{D}=s_\mu\times S_0$ and fully parametrized by $s_\mu$ and $s_\alpha$. In the following, we will use the more common parametrization in terms of the shape parameter $\alpha=s_\alpha$ and the scale parameter $\theta=s_\mu/s_\alpha$, where we have used that the factor of $s_\mu$ can be applied to the scale parameter.
Again, we show identifiability of time-dependent birth-death processes, not only for the constant-rate model considered in \cite{douglas_evolution_2025}.
\begin{theorem}
	Let $\hat{D}$ be gamma distributed with shape parameter $\alpha>0$ and scale parameter $\theta>0$. Furthermore, let $B$ be Poisson distributed with mean $\beta>0$ and $D=\sum_{i=0}^{B+1}\hat{D}_i$. Then all three parameters $\alpha$, $\theta$ and $\beta$ are identifiable from the shifted compound Poisson distribution~\eqref{eq:shifted_comp_poisson} of $D$.
\end{theorem}
\begin{proof}
The MGF of the gamma distribution of $\hat{D}$ is
\begin{align*}
    M_{\hat{D}}(t) &= (1-\theta t)^{-\alpha}
\end{align*}
and thus from \eqref{eq:cum_gen_D} we find the CGF of $D$ to be
\begin{align*}
    K(t) = -\alpha \log(1-\theta t) + \beta \left( (1-\theta t)^{-\alpha} - 1 \right),
\end{align*}
which exists for all $t<\theta^{-1}$ and has a pole at $t=\theta^{-1}$. From the uniqueness of the CGF for a given distribution and the identity theorem follows that $\theta$ can be uniquely inferred from the position of the pole.

This leaves us with showing that the distribution of $D$ is also unique for any value of $\alpha$ and $\beta$. For this, we again use the first three cumulants
\begin{align}
    \kappa_1 &= (\beta+1)\alpha\theta \label{eq:k1_gamma} \\
    \kappa_2 &= (\beta\alpha+\beta+1)\alpha\theta^2 \label{eq:k2_gamma}\\
    \kappa_3 &= \left(\beta\alpha^2+3\beta\alpha+2\beta+2\right)\alpha\theta^3, \label{eq:k3_gamma}
\end{align}
which can be derived by plugging the moments and cumulants of the gamma distribution
\begin{align*}
    m_{\hat{D},n} &= \theta^n\frac{\Gamma(\alpha+n)}{\Gamma(\alpha)} \\
    \kappa_{\hat{D},n} &= \theta^n\alpha(n-1)!
\end{align*}
into equation \eqref{eq:cumulant_D}.
Solving the equations for the first two cumulants of $D$ \eqref{eq:k1_gamma}-\eqref{eq:k2_gamma} for $\alpha$ and $\beta$ yields two solution
\begin{align*}
    \alpha_\pm &= \frac{1}{2\theta} \left( \kappa_1 \pm \sqrt{\kappa_1^2 + 4(\theta\kappa_1-\kappa_2)} \right) \\
    \beta_\pm &=  \frac{\alpha_\mp}{\alpha_\pm}.
\end{align*}
For the general case where these two solutions differ, they would yield the third cumulants \eqref{eq:k3_gamma} 
\begin{align*}
    \kappa_{3,\pm} &= \left( \alpha_\pm\alpha_\mp + 3\alpha_\mp + 2\frac{\alpha_\mp}{\alpha_\pm} + 2 \right)\alpha_\pm\theta \\
    &= -\left(\theta\kappa_1-\kappa_2\right)\alpha_\pm - 3\left(\theta\kappa_1-\kappa_2\right) +2\kappa_1.
\end{align*}
Suppose that both solutions, produce the same third cumulant, i.e. \ $\kappa_{3,+}=\kappa_{3,-}$, this would result in $\alpha_+=\alpha_-$ (because they also have the same first two cumulants) and therefore also in $\beta_+=\beta_-$. Hence, all parameters are uniquely identifiable from the first three cumulants and the CGF of $D$. This concludes the proof, as these quantities can be uniquely determined for a given distribution of $D$.
\end{proof}

In the original gamma spike model, evolutionary distance does not only accumulate from contributions of birth events, but additionally arises from a molecular clock. In the case of evolutionary distance and a homogeneous (strict) clock rate, this is a deterministic measure $c$ which is proportional to the branch length. The distribution of the observed sum $D_c$ of contributions to the evolutionary distance from bursts at birth events and the molecular clock is then given by
\begin{align*}
	P(D_c=d\mid \theta, M_\mathrm{burst},c) = \begin{cases}
		0 & c\leq d \\
	P(D=d-c\mid \theta, M_\mathrm{burst},c) & d>c.
	\end{cases}
\end{align*}
Thus, identifying the additional model parameter $c$ is trivial.

\section{Discussion}
Time-dependent birth-death processes are not identifiable from the phylogenetic tree alone \cite{louca_extant_2020,louca_fundamental_2021}. This has consequences for phylodynamic inference if no additional information on the population dynamics is available. Sequence alignments, however, in many cases yield more information than the tree alone. Here, we have considered the case where alterations to the genetic sequence occur at birth events. In this case, information on hidden birth events along branches of the phylogenetic tree can be obtained.

We have shown that time-dependent birth-death processes with either (or both) extant sampling and sampling through time are fully identifiable from the phylogenetic tree and the number of hidden birth events along branches. The latter does not need to be known exactly: For two different burst models, we have shown that evolutionary distance along branches yields sufficient information to infer the parameters of the burst model and the birth-death processes. The results can easily be extended to other burst models.

There are yet many open questions regarding the identifiability of generalized birth-death processes. Non-identifiabilities have only been shown for models with either extant sampling or sampling through time. The case where both sampling schemes are combined--as we do here--has, to the best of our knowledge, not been considered yet (but consult \cite{truman_fossilized_2025} for identifiability results of the fossilized birth-death process). Here we show that this case is identifiable if information on hidden birth events is available. An important next step will be to consider multi-type birth-death processes \cite{kuhnert_phylodynamics_2016} instead of single-type processes considered here. Such models are widely used in phylodynamics as they can account for heterogeneity in the population, e.g. by modeling different geographic locations with migration between them. The congruence classes of such processes have not been characterized yet.

Our results on the identifiability of time-dependent birth-death processes under models with accumulation of changes at birth have wide-ranging applications. Across phylodynamic application domains, ``birth'' events commonly coincide with an accumulation of mutations or substitutions: In single-cell biology, DNA is copied upon cell division so mutations occur at ``birth''. In epidemiology, upon transmission of a new host, the pathogen rapidly changes to adopt to this new host so mutations occur at ``birth''. In macroevolution, under the theory of punctuated equilibrium \cite{eldredge_punctuated_1972}, most changes occur at cladogenesis--meaning at ``birth''. Finally, in linguistics, there is evidence of bursts of change at ``birth'' \cite{atkinson_languages_2008} which has been hypothesized to be linked  to the widely considered process of schismogenesis \cite{bateson_199_1935,mansfield_linguistic_2025}. Thus our results have important consequences across application domains: birth and death rates are identifiable from merely the sequencing data.  

\subsubsection*{Acknowledgements}
We thank Marcus Overwater for helpful advice on Section~\ref{sec:summing_over_B} and Timothy Vaughan for proofreading. The project received funding from the European Research Council (ERC) under the European Union’s Horizon 2020 research and innovation programme grant agreement No 101001077 (PhyCogy).

\printbibliography

\appendix

\section{The number of hidden birth events in time-dependent birth-death processes}
The following derivations follow the approach in \cite{bokma_unexpectedly_2012}, where the distribution of the number of hidden birth events along branches was derived for the case of constant rates and extant sampling.

\begin{theorem}
	Consider a time-dependent birth-death process with extant sampling, defined by the parameters $\theta_{BD}=(\lambda(\tau),\mu(\tau),\rho)$, where $\lambda(\tau)$ and $\mu(\tau)$ are strictly positive and continuously differentiable functions of $\tau\geq 0$ and $\rho\in(0,1]$. The number of hidden birth events $B^{(\tau_s,\tau_e)}$ along a branch from $\tau_s$ to $\tau_e$ before the present ($\tau_s>\tau_e\geq0$) is Poisson distributed with mean
	\begin{align*}
		\mean{B^{(\tau_s,\tau_e)}\mid\theta_{BD}} = \int_{\tau_e}^{\tau_s}\d\tau 2\lambda(\tau)p_0(\tau),
	\end{align*}
	where $p_0(\tau)$ denotes the probability \eqref{eq:p0_extant} that an individual alive at time $\tau$ does not leave any sampled descendants.
	\label{th:poisson_B_extant_sampling}
\end{theorem}
\begin{proof}
As argued by \cite{bokma_unexpectedly_2012}, the rate with which hidden birth events arise along a branch at time $\tau$ before the present is given by $2\lambda(\tau)p_0(\tau)$, as birth event occur with rate $\lambda(\tau)$ and the probability of either offspring not being sampled is $2 p_0(\tau)$. The probability of an individual not undergoing any birth or death events between time $\tau_s$ and $\tau_e$ before the present is $\e^{-\int_{\tau_e}^{\tau_s}\d\tau(\lambda(\tau)+\mu(\tau))}$. Thus, the probability of hidden birth events occurring along a branch only at times $\tau_1,\dots,\tau_B$ is given by
\begin{align*}
	\e^{-\int_{\tau_e}^{\tau_s}\d\tau(\lambda(\tau)+\mu(\tau))}\frac{1}{B!}\prod_{i=1}^B 2\lambda(\tau_i)p_0(\tau_i).
\end{align*}
The factor $1/B!$ arises because the order of the birth events does not matter. Integrating over all possible times of the hidden birth events yields
\begin{align}
	\e^{-\int_{\tau_e}^{\tau_s}\d\tau(\lambda(\tau)+\mu(\tau))}\frac{1}{B!}\left(\int_{\tau_e}^{\tau_s}\d\tau2\lambda(\tau)p_0(\tau)\right)^B.
	\label{eq:prob_B_birth_events_uncond}
\end{align}
In order to condition on the existence of the branch from $\tau_s$ to $\tau_e$, we need to compute the probability $p_b(\tau_s,\tau_e)$ that an individual at time $\tau_s$ survives until time $\tau_e$ and all birth events in between do not lead to additional sampled descendants. This probability satisfies the master equation
\begin{align*}
	\partial_{\tau_s} p_b(\tau_s,\tau_e) = -(\lambda(\tau_s)+\mu(\tau_s))p_b(\tau_s,\tau_e) + 2\lambda(\tau_s)p_0(\tau_s)p_b(\tau_s,\tau_e).
\end{align*}
The first term corresponds to no event happening in an infinitesimal time interval before $\tau_s$, and the second term corresponds to the probability of a birth event from which one of the two descending clades remains unsampled. This master equation is identical to the master equation for the probability $p_1(\tau)$ that an individual alive at time $\tau$ leaves only a single sampled descendant
\begin{align}
	\partial_\tau p_1(\tau) = -(\lambda(\tau)+\mu(\tau))p_1(\tau) + 2\lambda(\tau)p_0(\tau)p_1(\tau).
    \label{eq:master_p1_extant_sampling}
\end{align}
However, both probabilities have different initial conditions $p_b(\tau_e,\tau_e)=1$ and $p_1(0)=\rho$. Thus, we find $p_b(\tau_s,\tau_e)=p_1(\tau_s)/p_1(\tau_e)$ (see \cite{stadler_decoding_2024} for extended derivations for the case of constant rates).
The probability mass function for the number of hidden birth events along a branch running from $\tau_s$ to $\tau_e$ is thus given by \eqref{eq:prob_B_birth_events_uncond} divided by $p_1(\tau_s)/p_1(\tau_e)$ to condition on the existence of the branch,
\begin{align}
	P(B^{(\tau_s,\tau_e)}=b\mid\theta_{BD}) = \e^{-\int_{\tau_e}^{\tau_s}\d\tau(\lambda(\tau)+\mu(\tau))}\frac{p_1(\tau_e)}{p_1(\tau_s)}\frac{1}{b!}\left(\int_{\tau_e}^{\tau_s}\d\tau2\lambda(\tau)p_0(\tau)\right)^b.
	\label{eq:poisson_B_extant_sampling}
\end{align}
To show that this is a Poisson distribution with mean $\int_{\tau_e}^{\tau_s}\d\tau2\lambda(\tau)p_0(\tau)$, we need to show that
\begin{align}
	-\log\left(\e^{-\int_{\tau_e}^{\tau_s}\d\tau(\lambda(\tau)+\mu(\tau))}\frac{p_1(\tau_e)}{p_1(\tau_s)}\right) = \int_{\tau_e}^{\tau_s}\d\tau2\lambda(\tau)p_0(\tau).
	\label{eq:show_poisson_mean}
\end{align}
We note that we can express the left-hand side as
\begin{align*}
	-\log\left(\e^{-\int_{\tau_e}^{\tau_s}\d u(\lambda(u)+\mu(u))}\frac{p_1(\tau_e)}{p_1(\tau_s)}\right) &= \left[ \int_{0}^{\tau}\d u(\lambda(u)+\mu(u)) + \log{p_1(\tau)} \right]_{\tau_e}^{\tau_s}
\end{align*}
Using the fundamental law of calculus, we are left with showing that the derivative of the term in square brackets is the integrand $2\lambda(\tau)p_0(\tau)$ in the right-hand side of \eqref{eq:show_poisson_mean}.
Using the master equation~\eqref{eq:master_p1_extant_sampling} for the derivative of $p_1(\tau)$, we get
\begin{align*}
	\partial_\tau \left[\int_{0}^{\tau}\d u(\lambda(u)+\mu(u)) + \log{p_1(\tau)}\right] &= \lambda(\tau)+\mu(\tau) + \frac{\partial_\tau p_1(\tau)}{p_1(\tau)} \\
	&= 2\lambda(\tau)p_0(\tau).
\end{align*}
\end{proof}

\begin{theorem}
	Consider a time-dependent birth-death process with sampling through time $\theta_{BDS}=(\lambda(\tau),\mu(\tau),\psi(\tau))$, where $\lambda(\tau)$, $\mu(\tau)$ and $\psi(\tau)$ are strictly positive and continuously differentiable functions of $\tau\geq 0$. The number of hidden birth events $B^{(\tau_s,\tau_e)}$ along a branch from $\tau_s$ to $\tau_e$ before the present ($\tau_s>\tau_e\geq0$) is Poisson distributed with mean
	\begin{align*}
		\mean{B^{(\tau_s,\tau_e)}\mid\theta_{BDS}} =\int_{\tau_e}^{\tau_s}\d\tau 2\lambda(\tau)p_0(\tau),
	\end{align*}
	where $p_0(\tau)$ is defined by \eqref{eq:diff_eq_p0}-\eqref{eq:init_p0} and denotes the probability that an individual alive at time $\tau$ does not leave any sampled descendants.
	\label{th:poisson_B_sampling_through_time}
\end{theorem}
\begin{proof}
	The proof is identical to the one for Theorem~\ref{th:poisson_B_extant_sampling}. The probability mass function for the number of hidden birth events \eqref{eq:poisson_B_extant_sampling} becomes
	\begin{align*}
		P(B^{(\tau_s,\tau_e)}=b\mid\theta_{BDS}) = \e^{-\int_{\tau_e}^{\tau_s}\d\tau(\lambda(\tau)+\mu(\tau)+\psi(\tau))}\frac{p_1(\tau_e)}{p_1(\tau_s)}\frac{1}{b!}\left(\int_{\tau_e}^{\tau_s}\d\tau2\lambda(\tau)p_0(\tau)\right)^b
	\end{align*}
	with the only difference being the addend of $\psi(\tau)$ in the exponent accounting for no sampling event during the time from $\tau_s$ to $\tau_e$. The master equation~\eqref{eq:master_p1_extant_sampling} for $p_1(\tau)$ becomes
	\begin{align*}
	\partial_\tau p_1(\tau) = -(\lambda(\tau)+\mu(\tau)-\psi(\tau))p_1(\tau) + 2\lambda(\tau)p_0(\tau)p_1(\tau).
	\end{align*}

\end{proof}

\section{Identifiability of constant-rate birth-death processes with extant sampling}
\label{sec:const_rates_derivation}
\begin{theorem}
	Let $\theta_{cBD}=(\lambda,\mu,\rho)$ and $\theta_{cBD}'=(\lambda',\mu',\rho')$ be two constant-rate birth-death models with extant sampling. If both models yield the same probability distributions of phylogenetic trees and the same statistics for the number of hidden birth events along branches of a tree, then they are identical.
\end{theorem}
\begin{proof}
The distribution of phylogenetic trees is the same for any two constant-rate birth-death processes with extant sampling $(\lambda,\mu,\rho)$ and $(\lambda',\mu',\rho')$ (i.e.\ they are congruent) if and only if 
\begin{align}
	\lambda'&=\frac{\rho}{\rho'}\lambda \label{eq:lambda_prime} \\
	\mu'&=\mu-\lambda(1-\frac{\rho}{\rho'}) \label{eq:mu_prime}
\end{align}
\cite{stadler_incomplete_2009}.

The number of hidden birth events along a branch running from $\tau_s$ to $\tau_e$ is Poisson distributed with mean
\begin{align*}
	\mean{B^{(\tau_s,\tau_e)}\mid\theta_{cBD}} &= \int_{\tau_e}^{\tau_s}\d\tau 2\lambda p_0(\tau) \\
	&= 2 \left[ \mu \tau - 
	\log\left(\lambda\rho + (\lambda(1-\rho)-\mu) \e^{-(\lambda-\mu)\tau}\right)\right]_{\tau_e}^{\tau_s}
\end{align*}
\cite{bokma_unexpectedly_2012,stolz_integrating_2024}. In the congruent model, the mean is
\begin{align*}
	\mean{B^{(\tau_s,\tau_e)}\mid\theta_{cBD}'} &=
	2 \left[ \mu' \tau - 
	\log\left(\lambda'\rho' + (\lambda'(1-\rho')-\mu') \e^{-(\lambda'-\mu')\tau}\right)\right]_{\tau_e}^{\tau_s} \\
    &= 2 \left[ \left(\mu-\lambda(1-\frac{\rho}{\rho'})\right)\tau - 
	\log\left(\lambda\rho + (\lambda(1-\rho)-\mu) \e^{-(\lambda-\mu)\tau}\right)\right]_{\tau_e}^{\tau_s} \\
	&= \mean{B^{(\tau_s,\tau_e)}\mid\theta_{cBD}} - 2 \lambda\left(1-\frac{\rho}{\rho'}\right)\left(\tau_s - \tau_e\right).
\end{align*}
This follows directly from the transformation rules above, as they lead to $\lambda'\rho'=\lambda\rho$, $\lambda'-\mu'=\lambda-\mu$ and $\lambda'(1-\rho')-\mu'=\lambda(1-\rho)-\mu$ (the same argument was used in~\cite{stadler_incomplete_2009}). The expected number of hidden birth events is thus the same for both processes only if $\rho=\rho'$. Due to \eqref{eq:lambda_prime}-\eqref{eq:mu_prime}, this would imply that both processes are identical, i.e.\ $(\lambda,\mu,\rho)=(\lambda',\mu',\rho')$.
\end{proof}

\end{document}